\documentclass[a4paper,twocolumn,11pt,unpublished]{quantumarticle}
\pdfoutput=1
\usepackage[utf8]{inputenc}
\usepackage[english]{babel}
\usepackage[T1]{fontenc}
\usepackage{amsmath}
\usepackage{amssymb}
\usepackage{amsthm}
\usepackage{algorithmic}

\newcounter{algorithm}
\makeatletter
\newenvironment{algorithm}[1][htbp]{%
  \par\vskip6pt
  \refstepcounter{algorithm}%
  \noindent\rule{\linewidth}{0.8pt}\par\vskip2pt
  \footnotesize
}{%
  \par\vskip2pt\noindent\rule{\linewidth}{0.8pt}%
  \par\vskip6pt
}
\newcommand{\AlgCaption}[1]{%
  \par\noindent{\raggedright\textbf{Algorithm~\thealgorithm}\enspace #1\par}\vskip2pt}
\makeatother
\usepackage{graphicx}
\usepackage{booktabs}
\usepackage[numbers,sort&compress]{natbib}
\usepackage{hyperref}

\newtheorem{theorem}{Theorem}
\newtheorem{definition}{Definition}

\newtheorem{proposition}{Proposition}

\begin{document}

\title{Distributed synthesis of arbitrary graph states in quantum networks via rank-two GF(2) reduction}
\author{Xiaoyi Zheng}
\email{xiaoyi.zheng@mpu.edu.mo}
\affiliation{%
 Faculty of Applied Sciences, Macao Polytechnic University, 999078, Macao SAR, China
}%

\author{Lin Chen}
\email{lchen@mpu.edu.mo}
\affiliation{%
 Faculty of Applied Sciences, Macao Polytechnic University, 999078, Macao SAR, China
}%

\author{Chan-Tong Lam}
\email{ctlam@mpu.edu.mo}
\affiliation{%
 Faculty of Applied Sciences, Macao Polytechnic University, 999078, Macao SAR, China
}%
\maketitle

\begin{abstract}
Existing schemes for synthesizing graph states in quantum networks are essentially edge-by-edge constructions, so quantities such as the time-slot depth and the resource overhead grow significantly with the edge density of the target graph. This paper proposes a new method. Exploiting the mathematical equivalence between joint Pauli-$X$ measurements and graph pivot operations, we formulate graph state synthesis as a rank-$2$ reduction process of a difference matrix over $\mathrm{GF}(2)$, and give an upper bound $\lfloor N/2\rfloor$ on the number of steps for synthesizing an arbitrary $N$-node graph state, independent of the edge density of the target graph state.
At the physical level, the joint Pauli-$X$ measurement of each step is mapped to a dual-star concurrent distribution.
We model the protocol on Waxman physical topologies with fiber attenuation and give a heuristic algorithm, and evaluate it against a strengthened Steiner baseline through Monte Carlo experiments.
The experimental results show that our protocol is superior in time-slot depth almost everywhere. The entanglement resource overhead, the total number of CZ gates, and the number of Pauli measurements drop below the baseline near edge density $p\approx0.3$, and are superior across the board thereafter. The denser the target graph state, the more significant the advantage.
\end{abstract}

\section{Introduction}
Graph states are a class of multipartite entangled states with a clear structure. Their entanglement structure can be represented by an undirected graph, and the corresponding evolution rules have been systematically studied~\cite{refgraph0,refgraph1,refgraph2,refgraph3}.
Synthesizing a specific graph state in a distributed network is mainly constrained by probabilistic entanglement generation and the limited coherence time of qubits. Established entanglement links must be used as soon as possible; the longer they wait, the larger the fidelity loss caused by decoherence~\cite{refhuang,refscalable}.
Therefore, the number of serial time slots is an important optimization objective for graph state synthesis schemes. Meanwhile, the fusion of multipartite entangled states inevitably introduces quantum operations such as CZ gates and Pauli measurements, which also affect the final fidelity of the graph state.

In this research area, Meignant et al.~\cite{refMMG} proposed a general scheme. The idea is to solve a Steiner tree on the physical links, synthesize star subgraphs along the Steiner tree via specific local Clifford operations, and then complete the distribution of an arbitrary graph state by stitching the star subgraphs together. When the target graph state is unknown, their approach first distributes an edge-decorated complete graph, and the nodes then project out the target state via local measurements. When the target graph state is known, one can extract the largest star subgraphs from the target graph, distribute them one by one, and stitch them together.
The Steiner tree strategy of the Meignant protocol~\cite{refMMG} reuses entanglement resources efficiently and is a clever application of classical graph theory in quantum networks.
However, in their idealized model, Meignant et al. assume lossless Clifford operations, instantaneous classical communication with zero delay, and simultaneous preparation of neighboring Bell pairs. In terms of time accounting, since the local Clifford operations required by the protocol commute with each other, the authors consider that a star graph state can be completed concurrently within one logical step. Consequently, the number of steps for synthesizing an arbitrary graph state is at most $N-1$.
This idealized assumption implicitly presumes that, within one time step, a single node can apply quantum operations in parallel to multiple internal qubits, i.e., quantum nodes possess strong parallel processing capability. This is difficult to realize, because too many parallel operations induce crosstalk between qubits.

Since then, researchers have carried out extensive work on distributing graph states in quantum networks~\cite{refgu2024fendi,refgao2025efficient,refliu2026swapping,refFan,refevan,refhuang,refckl,refscalable,reffischer,refSCsen}, along three main lines. The first optimizes the network layer. Fan et al.~\cite{refFan} studied maximizing the distribution rate under resource, decoherence, and fidelity constraints using hypergraph linear programming. Sutcliffe et al.~\cite{refevan} designed a fidelity-aware multi-path routing protocol. Huang et al.~\cite{refhuang} proposed the ST-P2PGSD protocol, which jointly optimizes resource allocation and path selection globally and provides a time accounting framework closer to physical reality.
The second line reduces the cost of the underlying primitive: the star graph state distribution schemes proposed in Refs.~\cite{refckl,refscalable} consume fewer quantum operations than Ref.~\cite{refMMG}. 
The third line moves the synthesis to a single node. The GST scheme proposed by Fischer et al.~\cite{reffischer} first prepares the target graph state at a single node and then teleports it to the target nodes; although efficient, it is essentially centralized distribution. The subgraph-complementation-based scheme of Sen et al.~\cite{refSCsen} outperforms Ref.~\cite{reffischer} in resource efficiency and introduces a gate-by-gate timing framework closer to physical reality. However, this protocol also relies on a central node, at which up to $N-1$ qubits must be aggregated, sharing the centralized architecture of the GST scheme, thereby creating a quantum resource bottleneck and limiting its scalability, particularly in current quantum architectures with limited local resources.

The above works have advanced multipartite entanglement distribution protocols along different dimensions such as network scheduling, algebraic synthesis, and low-level design. Most topology-aware distributed schemes, however, synthesize the target entanglement through edge- or star-based distribution primitives.
As the edge density $p$ increases, both the number of synthesis steps and the cumulative number of quantum gates inevitably grow, which limits their practicality for synthesizing dense graph states.

To address this problem, this paper revisits graph state synthesis in quantum networks from an algebraic perspective, mapping it to pivot transformations on nodes of a distributed quantum network.
Specifically, we represent a graph state by its adjacency matrix over $\mathrm{GF}(2)$, and a joint Pauli-$X$ measurement is interpreted as one symmetric rank-$2$ update of this matrix.
At the concrete physical level, each rank-$2$ update is mapped to the concurrent distribution of three entanglement branches: the backbone entanglement between the ancilla pair $S_A$ and $S_B$, and the star-type entanglement distribution from $S_A$ and from $S_B$ to the target nodes.
We further prove that synthesizing an arbitrary graph state is equivalent to performing a sequence of rank-$2$ updates on the difference matrix until it vanishes, and that the number of steps for synthesizing an arbitrary $N$-node graph state is upper bounded by $\lfloor N/2\rfloor$; this upper bound is determined by the rank of the adjacency matrix of the target graph state, as stated in Theorem~\ref{thm:complexity}.

This algebraic viewpoint is consistent with a body of recent results. 
Burchardt et al. ~\cite{refburchardt2025} introduced foliage partitions as efficiently computable invariants of LC equivalence classes of graph states,
Sharma et al.~\cite{refsharma2026} used integer linear programming to search within an equivalence class for the equivalent graph state with the fewest edges, 
and L\"obl et al.~\cite{refmatthias2025} derived graph-transformation rules for Bell-state measurements on graph states.
These results show that there is a clear algebraic correspondence between the topological evolution of graph states and quantum operations. Oum~\cite{refoum} developed the vertex-minor and rank-width framework in which graph pivoting plays an important role, and pivoting remains an active tool in structural graph theory ~\cite{refoum2025}. Mhalla and Perdrix~\cite{refmhalla} proved the physical equivalence between joint Pauli-$X$ measurements and pivot operations, and Van den Nest et al.~\cite{refmaarten} characterized the graph transformations induced by local Clifford operations on graph states.

It should be noted that Huang et al.~\cite{refhuang} proved that, when node storage is unlimited, the number of time slots required for a one-shot distribution of a given graph state has a tight upper bound $\lfloor|V_S|/2\rfloor$, numerically similar to our $\lfloor N/2\rfloor$; however, the two count different objects. The former counts the number of rounds of parallel routing, whereas ours counts the number of rank-$2$ operations with the joint Pauli-$X$ measurement as the primitive. Their meanings are not the same.
In mathematical spirit, our work is closest to the subgraph complementation approach of Sen et al.~\cite{refSCsen}; both originate from algebraic transformations over $\mathrm{GF}(2)$. 
The former is based on subgraph-complementation systems closely related to the minimum-rank representation of the graph, whereas our pivot transformation is expressed directly as symmetric rank-$2$ updates of the adjacency matrix.
Consequently, for random graphs the subgraph complementation approach requires a number of steps close to $N$, whereas our step count is $\operatorname{rank}(M_{\mathrm{target}})/2$; since the standard rank is close to $N$ for random graphs, this is about $N/2$, only half as many. More importantly, the protocol of Sen et al. is essentially a centralized protocol, requiring up to $N-1$ qubits to be aggregated at a central node, and it does not address a physical implementation under a distributed execution architecture.

To evaluate the performance of our protocol, we do not adopt an idealized time accounting model; instead, we evaluate the protocol and the baseline under a time-slot model close to physical conditions. The time-slot model follows the framework given in Refs.~\cite{refshiqian,refhuang}. That is, time is discretized into slots of fixed length, classical information can traverse the entire network within one slot, and neighboring nodes attempt to establish Bell pairs according to the channel capacity.
On this basis, we add a single-interface constraint: each node can participate in the establishment of only one entanglement link per time slot, which avoids crosstalk from parallel operations within a node.
Consequently, physical entanglement links sharing the same quantum node must be staggered into different slots and executed serially. For the numerical comparison, the maximum node degree in each step is used as a degree-based slot-depth measure, and this measure is summed over all steps.
In this timing framework, classical communication is not a bottleneck: measurement outcomes can be sent point-to-point directly, and Pauli byproducts are recorded in a Pauli frame and corrected at the end.
As the comparison baseline, we take Ref.~\cite{refMMG} as the skeleton and strengthen it where necessary: the star subgraph generated in each step adopts the scheme with lower gate cost from Refs.~\cite{refckl,refscalable}, and the entanglement establishment follows its Steiner tree idea. Below we refer to this as the Steiner baseline.
We map the proposed method onto a concrete network model, give a heuristic algorithm that handles pivot ordering and ancilla placement separately, and build the physical network model on Waxman topologies with attenuation. We validate our protocol and the Steiner baseline via Monte Carlo numerical experiments.
The experimental results show that, since the number of serial steps of our protocol is fixed and small, it is superior in the degree-based time-slot metric almost everywhere. The entanglement resource overhead, the total number of CZ gates, and the number of Pauli measurements drop below the baseline near edge density $p\approx0.3$, and are superior across the board thereafter. The denser the target graph state, the more significant the advantage.
The paper is organized as follows. Section~\ref{sec:2} reviews the algebraic representation of graph states, Sec.~\ref{sec:3} establishes the rank-$2$ synthesis framework, Sec.~\ref{sec:network.model} completes the physical mapping and gives the heuristic algorithm, Sec.~\ref{sec:numerical} reports the numerical results and Sec.~\ref{sec:conclusion} concludes.

\section{Preliminaries}
\label{sec:2}
This section first reviews the definition of graph states and their adjacency matrix representation over $\mathrm{GF}(2)$, and introduces the basic operations and evolution rules of graph states~\cite{refgraph0,refgraph1,refgraph2,refgraph3}. On this basis, we clarify the equivalence between joint Pauli-$X$ measurements and the pivot transformation in graph theory.

\subsection{Graph states and their GF(2) algebraic representation}

Table~\ref{tab:notation} collects the notation used throughout the paper.

\begin{table}[htbp]
\centering
\footnotesize
\setlength{\tabcolsep}{4pt}
\renewcommand{\arraystretch}{1.05}
\begin{tabular}{@{}ll@{}}
\hline
Symbol & Meaning \\
\hline
$N$ & number of target nodes \\
$M$ & GF(2) adjacency matrix of a graph state \\
$R$ & difference matrix $M_{\mathrm{target}}\oplus M_{\mathrm{curr}}$ \\
$e_i$,\ $Re_i$ & $i$-th standard basis vector; $i$-th column of $R$ \\
$u,v$ & ancilla connection vectors, $u=Re_j$, $v=Re_i$ \\
$S_A,S_B$ & entangled ancilla pair of one rank-2 update \\
$G\wedge uv$ & pivot of $G$ along the edge $uv$ \\
$\oplus$ & bitwise XOR, i.e.\ addition over GF(2) \\
$C_e$ & resource cost of physical edge $e$ \\
$\Delta_{\mathrm{net}}$,\ $\Delta(G)$ & net edge progress; max node degree \\
\hline
\end{tabular}
\caption{Notation used throughout the paper.}
\label{tab:notation}
\end{table}

A graph state $|G\rangle$ is a class of multipartite entangled states whose entanglement structure corresponds one-to-one to a simple undirected graph $G=(V,E)$ without self-loops. To prepare it, each vertex is assigned a qubit initialized to $|+\rangle = (|0\rangle + |1\rangle)/\sqrt{2}$, and then a controlled-Z gate $\mathrm{CZ}_{a,b}$ is applied for each edge $(a,b)\in E$ to establish entanglement between the corresponding qubits. The entanglement structure of the resulting state is fully isomorphic to the topology of the graph $G$:
\begin{equation}
|G\rangle = \prod_{(a,b)\in E} \mathrm{CZ}_{a,b} |+\rangle^{\otimes |V|}.
\end{equation}

An $N$-qubit graph state directly corresponds to an adjacency matrix $M \in \mathrm{GF}(2)^{N \times N}$ over the finite field $\mathrm{GF}(2)$: if a logical edge exists between nodes $i$ and $j$, then $M_{i,j}=M_{j,i}=1$, and otherwise $0$. Since the graph has no self-loops, the main diagonal is identically zero, so the adjacency matrix is symmetric with an all-zero diagonal, which means that $M$ is an alternating matrix over $\mathrm{GF}(2)$. 
Over $\mathrm{GF}(2)$, the rank of an alternating matrix is always even, so the initial rank $\operatorname{rank}(M_{\mathrm{target}})$ must be even. Unless stated otherwise, all matrix arithmetic in this paper is carried out over $\mathrm{GF}(2)$, so the reduction modulo $2$ is left implicit below.

\subsection{Basic quantum operations and evolution rules of graph states}

Local complementation (LC) and single-qubit Pauli measurements are two classes of basic operations on graph states. LC changes the adjacency relations between the nodes of a graph state, but the resulting graph state remains local-Clifford equivalent to the original one; Pauli measurements change the topology of the graph state and are usually accompanied by byproducts, which require classical feedback corrections.

\begin{itemize}
\item \textbf{Local complementation (LC).} Applying $\mathrm{LC}_i(G)$ to vertex $i$ of a graph state $G=(V,E)$ is equivalent to taking the symmetric difference $G \mapsto (V, E \triangle \mathcal{E}_i)$ on all edges within its neighborhood $\mathcal{N}(i)$, where $\mathcal{E}_i = \{ (u,v) \mid u,v \in \mathcal{N}(i), u \neq v \}$. Physically, it corresponds to the following local Clifford unitary,
\begin{equation}
U_{\mathrm{LC}_i} = R_x^{(i)}(\pi/2) \bigotimes_{j \in \mathcal{N}(i)} R_z^{(j)}(-\pi/2).
\end{equation}
\item \textbf{Pauli measurement rules.} A single-qubit Pauli measurement on node $i$ induces a global evolution of the graph state. Let $G_{-i}$ denote the residual graph after removing $i$. Depending on the chosen measurement basis, after local Clifford corrections the remaining system evolves into a new graph state,

\begin{equation}
\label{eq:pauli_measurements}
\begin{aligned}
M_z^{(i)} |G\rangle &= \tfrac{1}{\sqrt{2}} \sum_{k=\pm 1} |k\rangle_i \otimes U_k^{(i)} |G_{-i}\rangle, \\
M_y^{(i) } |G\rangle &= \tfrac{1}{\sqrt{2}} \sum_{l=\pm 1} |l\rangle_i \otimes U_l^{(i)} |\mathrm{LC}_i(G)_{-i}\rangle, \\
M_x^{(i)} |G\rangle &= \tfrac{1}{\sqrt{2}} \sum_{m=\pm 1} |m\rangle_i \\
&\qquad \otimes U_m^{(i,j)} |\mathrm{LC}_j(\mathrm{LC}_i(\mathrm{LC}_j(G))_{-i})\rangle.
\end{aligned}
\end{equation}

Here $|k\rangle_i, |l\rangle_i, |m\rangle_i$ are the measurement eigenstates of the corresponding Pauli bases, $j\in\mathcal{N}(i)$ is an arbitrarily chosen neighbour of $i$, and $U_k^{(i)},U_l^{(i)},U_m^{(i,j)}$ are outcome-dependent byproduct operators. If they commute with subsequent quantum operations, they can be recorded in a Pauli frame and corrected afterwards.
\end{itemize}

\subsection{Joint Pauli-X measurement and the pivot transformation}

According to the Pauli-$X$ rule in Eq.~\eqref{eq:pauli_measurements}, the residual graph produced by a single-qubit Pauli-$X$ measurement depends on the choice of the reference vertex, and the graphs corresponding to different reference vertices are LC-equivalent to each other. Mhalla and Perdrix~\cite{refmhalla} proved that performing Pauli-$X$ measurements simultaneously on adjacent vertices $u,v$ (hereafter called a joint Pauli-$X$ measurement) is equivalent to first performing the pivot operation $G \wedge uv$ on the original graph and then deleting the two measured vertices:
\begin{equation}
\label{eq:joint_X_measurement}
\langle s_u^{(\pi/2)} r_v^{(\pi/2)} | G \rangle = \tfrac{1}{2} Z_S | G \wedge uv \setminus \{u,v\} \rangle,
\end{equation}
where $s,r \in \{0,1\}$ are the measurement outcomes, which determine the byproduct correction set $S$:
\begin{equation}
S = (\mathcal{N}(u) \cap \mathcal{N}(v)) \triangle (\mathcal{N}(u) \setminus \{v\})^r \triangle (\mathcal{N}(v) \setminus \{u\})^s.
\end{equation}
Since Pauli-$Z$ byproducts do not change the edge set of the graph, once these local phases are compensated by classical feedforward, the effect of the joint Pauli-$X$ measurement on the graph topology is uniquely determined. This conclusion can also be derived from the Pauli-$X$ measurement rule in Eq.~\eqref{eq:pauli_measurements}. When $u,v$ are adjacent, first choose a neighbor $w$ of $u$ as the reference vertex, giving $X_u \equiv Z_u\, \mathrm{LC}_w \circ \mathrm{LC}_u \circ \mathrm{LC}_w \equiv Z_u (G \wedge wu)$; after the measurement, $w$ necessarily becomes a neighbor of $v$~\cite{refgraph3,refshen}, so taking $w$ as the reference again gives $X_v \equiv Z_v (G \wedge wv)$. Hence:
\begin{equation}
X_u X_v \equiv Z_u Z_v\, (G \wedge wu) \wedge wv.
\end{equation}
By Proposition 2.5 of Oum~\cite{refoum}, consecutive pivots satisfy $G \wedge wu \wedge wv = G \wedge uv$, and substituting this yields
\begin{equation}
X_u X_v \equiv Z_u Z_v\, (G \wedge uv).
\end{equation}
Therefore, the conclusion of Eq.~\eqref{eq:joint_X_measurement} is fully consistent with ``two single-qubit Pauli-$X$ measurements taking the same reference vertex''. The joint Pauli-$X$ measurement is equivalent to the graph pivot transformation at the topological level. In the next section, we rewrite this rule as an update of the $\mathrm{GF}(2)$ adjacency matrix, thereby giving the synthesis problem an algebraic formulation.

\section{Algebraic framework for rank-2 graph state synthesis}
\label{sec:3}
The previous section clarified the equivalence rule of the joint Pauli-$X$ measurement at the topological level. This section establishes its algebraic formulation: the graph state is represented by its $\mathrm{GF}(2)$ adjacency matrix, the effect of the joint measurement is written as one matrix update, and graph state synthesis is thereby reduced to a step-by-step rank reduction of a difference matrix. In this formulation, the number of synthesis steps is directly determined by the rank of the matrix.

\subsection{Algebraic update rule of graph state evolution}

According to Theorem 3.1 of Kim and Oum~\cite{kim2024vertex} and Tucker's matrix pivoting theory~\cite{reftucker1963,refPPT2000}, the pivot operation $G \wedge S_A S_B$ on a graph is algebraically equivalent to one principal pivot elimination of the binary adjacency matrix. We now expand this process and give its explicit form.

Let $M$ be the adjacency matrix of the entire undirected graph. Placing the adjacent node pair $S_A,S_B$ to be measured in the first two positions, $M$ has the block form
\begin{equation}
M =
\begin{pmatrix}
\alpha & \beta \\
\gamma & \delta
\end{pmatrix} ,
\end{equation}
where $\alpha$ is the $2 \times 2$ principal adjacency submatrix between $S_A$ and $S_B$. Since the graph has no self-loops and $S_A,S_B$ are adjacent, $\alpha = \left(\begin{smallmatrix} 0 & 1 \\ 1 & 0 \end{smallmatrix}\right)$, which is invertible and satisfies $\alpha^{-1}=\alpha$. The blocks $\beta = \big(\begin{smallmatrix} u^{\mathrm{T}} \\ v^{\mathrm{T}} \end{smallmatrix}\big)$ and $\gamma = \beta^{\mathrm{T}} = (u\ v)$ record the connections between $S_A,S_B$ and the remaining nodes; the column vectors $u,v$ describe the edges from the remaining nodes to $S_A$ and to $S_B$, respectively. The block $\delta$ is the principal submatrix remaining after deleting $S_A,S_B$.

The principal pivot elimination rule gives the update of the residual submatrix $\delta' = \delta - \gamma \alpha^{-1} \beta$. Over $\mathrm{GF}(2)$, addition and subtraction coincide and are both bitwise XOR $\oplus$; expanding the matrix product gives
\begin{equation}
\gamma \alpha^{-1} \beta = (u\ v) \begin{pmatrix} 0 & 1 \\ 1 & 0 \end{pmatrix} \begin{pmatrix} u^{\mathrm{T}} \\ v^{\mathrm{T}} \end{pmatrix} = u v^{\mathrm{T}} + v u^{\mathrm{T}}  .
\end{equation}
This yields the algebraic update rule of the joint Pauli-$X$ measurement.

\begin{proposition}[Rank-2 update rule of the joint measurement]
\label{prop:rank2}
Let $M$ be the adjacency matrix of a graph state over $\mathrm{GF}(2)$. After performing a joint Pauli-$X$ measurement on adjacent nodes $(S_A,S_B)$ and compensating the corresponding local Clifford corrections, the adjacency matrix of the remaining subsystem is updated to
\begin{equation}
\label{eq:pivoting_matrix}
M' = M_{\setminus \{S_A,S_B\}} \oplus (u v^{\mathrm{T}} + v u^{\mathrm{T}}) .
\end{equation}
\end{proposition}

Proposition~\ref{prop:rank2} shows that one joint measurement is algebraically equivalent to XOR-ing the outer-product term $uv^{\mathrm{T}} + vu^{\mathrm{T}}$ with the adjacency matrix of the original graph. The rank of this outer-product term is at most $2$, and in our construction it is exactly $2$ (see the proof of Theorem~\ref{thm:complexity}); hence we call it a rank-$2$ update. Physically, it corresponds to the concurrent flipping of those logical edges determined by the vectors $u$ and $v$. 
Symmetric rank-$2$ GF(2) perturbations of this form also serve as a combinatorial tool in the study of vertex-minors~\cite{refoum2025}.
Related GF(2) rank-based decompositions have also been studied via odd covers~\cite{refbuchanan2023odd}, whereas here the rank-$2$ update $uv^{T}+vu^{T}$ is realized physically through a joint Pauli-$X$ measurement.

\subsection{Difference matrix}

\begin{definition}[Difference matrix]
Define the difference matrix as the difference between the target adjacency matrix $M_{\mathrm{target}}$ and the current adjacency matrix $M_{\mathrm{curr}}$:
\begin{equation}
R = M_{\mathrm{target}} \oplus M_{\mathrm{curr}}.
\end{equation}
\end{definition}

Suppose the target network has $N$ nodes. Introduce a pair of entangled ancilla qubits $S_A,S_B$, and denote their connection states to the original $N$ nodes by the column vectors $u,v\in\mathrm{GF}(2)^N$. Placing the ancilla nodes in the first two positions of the matrix, the extended global system (with $N+2$ nodes in total) has the block form:
\begin{equation}
M_{\mathrm{total}} =
\begin{pmatrix}
0 & 1 & u^{\mathrm{T}} \\
1 & 0 & v^{\mathrm{T}} \\
u & v & M_{\mathrm{curr}}
\end{pmatrix} .
\end{equation}
The upper-left $\left(\begin{smallmatrix}0&1\\1&0\end{smallmatrix}\right)$ block indicates that an entanglement edge already exists between $S_A$ and $S_B$, which is the prerequisite for the joint measurement. By Proposition~\ref{prop:rank2}, after performing a joint Pauli-$X$ measurement on $(S_A,S_B)$, the topology of the target network evolves as:
\begin{equation}
M_{\mathrm{curr}}' = M_{\mathrm{curr}} \oplus (u v^{\mathrm{T}} + v u^{\mathrm{T}}).
\end{equation}

Initially, the network has no entanglement, $M_{\mathrm{curr}} = \mathbf{0}$, hence $R^{(0)} = M_{\mathrm{target}}$. Our goal is to construct suitable vectors $(u,v)$ so that $M_{\mathrm{curr}}$ approaches $M_{\mathrm{target}}$ step by step, i.e., $R \to \mathbf{0}$. Graph state synthesis thus reduces to the problem of nullifying the difference matrix $R$ over $\mathrm{GF}(2)$ step by step.

\subsection{Rank-2 nullification theorem and complexity upper bound}

The key to the nullification lies in the construction of the ancilla connection vectors $u,v$. Theorem~\ref{thm:nullification} gives a construction of the vectors $u,v$ such that one joint Pauli-$X$ measurement simultaneously nullifies two rows and two columns of $R$.

\begin{theorem}[Rank-2 nullification]
\label{thm:nullification}
Suppose the current difference matrix $R$ has a nonzero element $R_{i,j}=1$. If we set $u = R e_j$ and $v = R e_i$, then after performing the joint Pauli-$X$ measurement on this ancilla pair, the updated difference matrix $R'$ satisfies
\[
R' e_i = \mathbf{0}, \quad R' e_j = \mathbf{0}.
\]
\end{theorem}
\begin{proof}
By Proposition~\ref{prop:rank2} and the definition of the difference matrix, the evolution rule is $R' = R \oplus (u v^{\mathrm{T}} + v u^{\mathrm{T}}) $. Multiplying by the basis vector $e_i$ on the right:
\[
R' e_i = R e_i \oplus u (v^{\mathrm{T}} e_i) \oplus v (u^{\mathrm{T}} e_i).
\]
From $v = R e_i$, we have $v^{\mathrm{T}} e_i = e_i^{\mathrm{T}} R^{\mathrm{T}} e_i = R_{i,i}$. The diagonal of the difference matrix is identically zero, so $R_{i,i}=0$ and this term vanishes. From $u = R e_j$, we have $u^{\mathrm{T}} e_i = e_i^{\mathrm{T}} R^{\mathrm{T}} e_j = R_{j,i}$, and by symmetry $R_{j,i}=R_{i,j}=1$, so
\[
R' e_i = R e_i \oplus \mathbf{0} \oplus v \cdot 1 = R e_i \oplus v = v \oplus v = \mathbf{0}.
\]
The claim $R' e_j = \mathbf{0}$ follows in the same way.
\end{proof}

Theorem~\ref{thm:nullification} states that, as long as $u$ and $v$ are chosen according to the difference matrix, one joint measurement nullifies two rows and two columns of $R$. Iterating this procedure gradually restores the target graph state. Theorem~\ref{thm:complexity} further gives the theoretical upper bound on the number of steps required by this iterative process.

\begin{figure*}[t]
    \centering
    \includegraphics[width=0.9\textwidth]{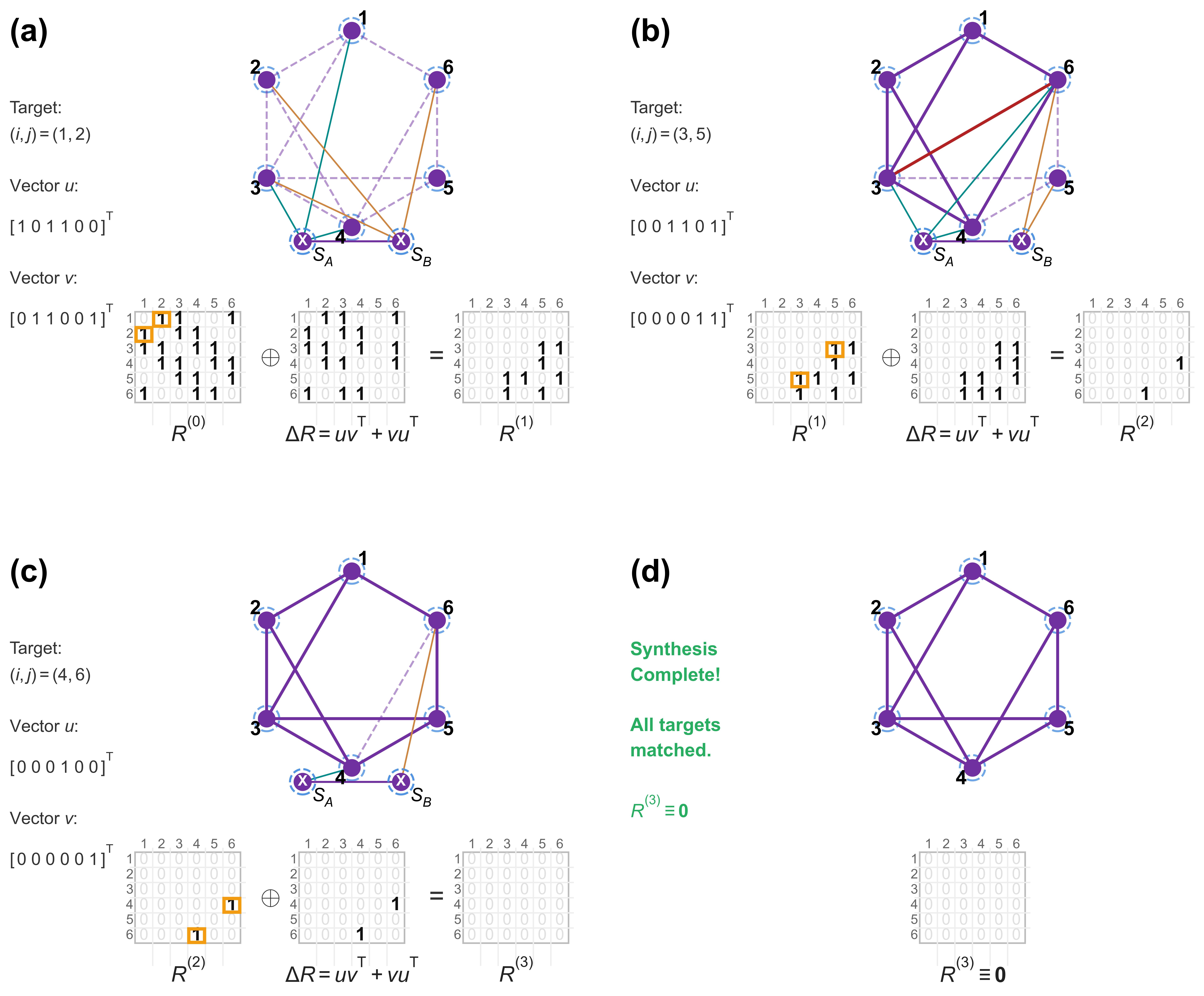}
    \caption{Step-by-step graph state synthesis based on rank-2 updates. In subfigures (a) to (c), the upper part is the physical graph and the lower part is the corresponding GF(2) difference matrix $R^{(k)}$; the two evolve synchronously. In each step, a nonzero pivot (golden square) is selected in $R^{(k)}$, from which the vectors $u,v$ are extracted to configure the ancilla qubits $S_A,S_B$ and perform the joint Pauli-$X$ measurement. Algebraically, this measurement is the rank-2 update $\Delta R=uv^{\mathrm{T}}+vu^{\mathrm{T}}$; physically, it generates the target edges while concurrently producing the redundant edges marked in red. In subfigure (d), the final state satisfies $R^{(3)}\equiv\mathbf{0}$, and the purple synthesized graph coincides exactly with the target state.}
    \label{fig:main_evolution}
\end{figure*}

\begin{theorem}[Upper bound on the number of global operations]
\label{thm:complexity}
Suppose a system of $N$ nodes is initially in an unentangled state, and the target graph state corresponds to the adjacency matrix $M_{\mathrm{target}}$. By iterating the measurement operation of Theorem~\ref{thm:nullification}, the number of operations $K$ required to evolve the system to the target graph state is exactly
\[
K = \frac{\operatorname{rank}(M_{\mathrm{target}})}{2},
\]
and this number of operations satisfies the global upper limit $K \le \left\lfloor \frac{N}{2} \right\rfloor$.
\end{theorem}

\begin{proof}
The system is initially unentangled, so $R^{(0)}=M_{\rm target}$. Let $R^{(k-1)}$ be the difference matrix before the $k$-th operation, and choose a nonzero element $R_{i,j}^{(k-1)} = 1$ for the operation.

By the update rule of Theorem~\ref{thm:nullification}, the difference matrices before and after the operation are related by
\[
R^{(k)} = R^{(k-1)} \oplus \bigl(u v^{\mathrm{T}} + v u^{\mathrm{T}}\bigr).
\]
The update term $u v^{\mathrm{T}} + v u^{\mathrm{T}}$ is symmetric with an identically zero diagonal (its $l$-th diagonal element is $2u_lv_l\equiv 0 \pmod 2$), so $R^{(k)}$ remains an alternating matrix over $\mathrm{GF}(2)$, and the premise $R_{ii}=0$ required by Theorem~\ref{thm:nullification} holds at every step.

Since $\oplus$ is self-inverse over $\mathrm{GF}(2)$, we also have $R^{(k-1)} = R^{(k)} \oplus \bigl(u v^{\mathrm{T}} + v u^{\mathrm{T}}\bigr)$. The rank of the perturbation term is at most $2$; by the subadditivity of matrix rank, $\mathrm{rank}(A\oplus B)\le \mathrm{rank}(A)+\mathrm{rank}(B)$, we obtain $\mathrm{rank}(R^{(k-1)}) \le \mathrm{rank}(R^{(k)}) + 2$, i.e., the rank drops by at most $2$ per step: $\mathrm{rank}(R^{(k)}) \ge \mathrm{rank}(R^{(k-1)}) - 2$.

On the other hand, by the update formula, every column of $R^{(k)}$ belongs to $\operatorname{span}\bigl(\operatorname{col}R^{(k-1)}\cup\{u,v\}\bigr)=\operatorname{col}R^{(k-1)}$, because $u=R^{(k-1)}e_j$ and $v=R^{(k-1)}e_i$ are themselves columns of $R^{(k-1)}$. Moreover, by Theorem~\ref{thm:nullification} and symmetry, the $i$-th and $j$-th rows of $R^{(k)}$ are all zero, so every vector in $\operatorname{col}R^{(k)}$ vanishes at coordinates $i$ and $j$. But $u_i=R^{(k-1)}_{i,j}=1$, $v_j=R^{(k-1)}_{i,j}=1$, $v_i=R^{(k-1)}_{i,i}=0$, so $u$, $v$, and $u\oplus v$ all take the value $1$ at coordinate $i$ or $j$, and none of them belongs to $\operatorname{col}R^{(k)}$; that is, $\{u,v\}$ remains linearly independent when combined with $\operatorname{col}R^{(k)}$. Therefore $\operatorname{col}R^{(k-1)}\supseteq\operatorname{col}R^{(k)}\oplus\operatorname{span}\{u,v\}$, and the rank drops by at least $2$: $\operatorname{rank}\bigl(R^{(k)}\bigr) \le \operatorname{rank}\bigl(R^{(k-1)}\bigr) - 2$.

Combining the upper and lower bounds, each step reduces $\mathrm{rank}(R)$ by exactly 2:
\[
\operatorname{rank}\bigl(R^{(k)}\bigr) = \operatorname{rank}\bigl(R^{(k-1)}\bigr) - 2.
\]

Over $\mathrm{GF}(2)$, the rank of an alternating matrix is always even, so the initial rank $\operatorname{rank}(M_{\mathrm{target}})$ must be even. The rank sequence starts from $\operatorname{rank}(M_{\mathrm{target}})$ and strictly decreases by 2 at each step until it reaches 0, completing the synthesis; hence the total number of operations required is exactly $K = \operatorname{rank}(M_{\mathrm{target}})/2$. Since $M_{\mathrm{target}}$ is an $N \times N$ matrix, its rank is at most $N$ and is even, so $K \le \lfloor N/2 \rfloor$.
\end{proof}

The step-count upper bound $\lfloor N/2\rfloor$ given by Theorem~\ref{thm:complexity} is independent of the edge density of the target graph and depends only on the standard rank of its adjacency matrix. Based on the above two theorems, we obtain the purely algebraic graph state synthesis Algorithm~\ref{alg:algebraic}. Figure~\ref{fig:main_evolution} demonstrates the matrix evolution of this algorithm on a 6-node target graph state.

\begin{algorithm}[htbp]
\AlgCaption{Algebraic graph state synthesis over GF(2)}
\label{alg:algebraic}
\begin{algorithmic}[1]
\REQUIRE Target graph state adjacency matrix $M_{\mathrm{target}}$, available ancilla node set $\mathcal{V}_{\mathrm{ancilla}}$
\ENSURE Sequence of algebraic operations $\mathcal{S}$ achieving the target graph state
\STATE $M_{\mathrm{curr}} \leftarrow \mathbf{0}$, $R \leftarrow M_{\mathrm{target}}$, $\mathcal{S} \leftarrow \emptyset$
\WHILE{$R \neq \mathbf{0}$}
    \STATE Choose any nonzero element $R_{i,j}=1$ of $R$
    \STATE $u \leftarrow R e_j$, $v \leftarrow R e_i$
    \STATE Take a pair of ancilla nodes $(S_A,S_B)$ sharing entanglement from $\mathcal{V}_{\mathrm{ancilla}}$
    \STATE $\mathcal{S} \leftarrow \mathcal{S} \cup \{(u,v,S_A,S_B)\}$
    \STATE $R \leftarrow R \oplus (u v^{\mathrm{T}} + v u^{\mathrm{T}})$
\ENDWHILE
\RETURN $\mathcal{S}$
\end{algorithmic}
\end{algorithm}

Each tuple $(u,v,S_A,S_B)$ fully specifies one synthesis operation: the ancilla pair $(S_A,S_B)$ is initialized in a Bell state, connected to the target nodes indicated by the supports of $u$ and $v$, and then jointly measured in the Pauli-$X$ basis. All constituent operations, namely CZ gates, Pauli measurements, and byproduct corrections, are Clifford operations; their physical realization is given in Sec.~\ref{subsec:network_model}.

At each step, Algorithm~\ref{alg:algebraic} extracts $u,v$ from the difference matrix so that the rank drops by $2$, and the number of synthesis steps does not exceed $\lfloor N/2\rfloor$, whereas the number of steps of the Steiner baseline equals the number of decomposed star subgraphs, which grows as the target graph becomes denser. Under the time-slot model, the physical cost of a single step is limited by the physical connectivity, and the two protocols show no essential difference in order, although the maximum degree of the dual-star routing graph is somewhat higher; therefore, the difference in cumulative depth mainly originates from the number of serial steps. This point will be discussed in Sec.~\ref{sec:concurrent} and numerically in Sec.~\ref{subsec:time_cost}.
Note that Algorithm~\ref{alg:algebraic} is a purely idealized algebraic algorithm that does not yet account for the spatial distances and channel losses in a real physical network. In a real quantum network, the choice of which edge to eliminate at each step and where to place the ancilla nodes $(S_A,S_B)$ can significantly affect the physical resources consumed for establishing the underlying connections. In the next section, we map the algebraic algorithm onto a real physical network and design a heuristic algorithm under this physical mapping with the goal of optimizing the resource overhead.

\section{Network model and physical mapping strategy}
\label{sec:network.model}

In this section, we first present the quantum network model and three basic evolution operations, define the physical implementation of one rank-$2$ update as a dual-star concurrent distribution, and compute its single-step depth under the time-slot model. On this basis, taking into account the channel attenuation faced by entanglement establishment in real networks, we design a heuristic strategy algorithm with the goal of minimizing the resource overhead.

\begin{figure*}[htbp]
    \centering
    \includegraphics[width=1\textwidth]{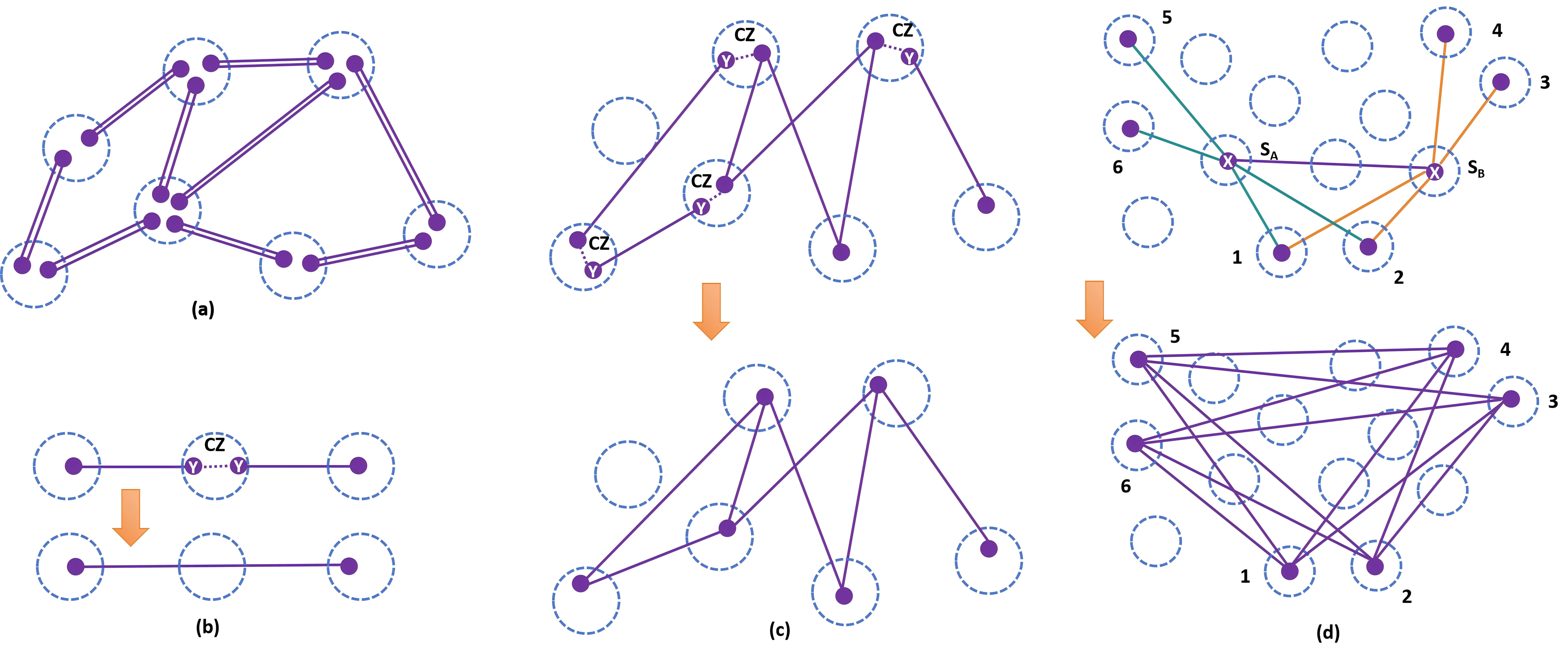}
    \caption{Quantum network model and basic operations. Subfigure (a) shows the physical network topology, where double solid lines are the quantum channels between neighbouring nodes. Subfigure (b) shows the establishment of long-range entanglement: entanglement swapping via local CZ gates and Pauli-$Y$ measurements cascades short-range Bell pairs into a long-range entanglement edge. Subfigure (c) shows state stitching: a newly generated entanglement edge is attached to an existing graph state node via local CZ gates and Pauli-$Y$ measurements, preparing for subsequent topology evolution. Subfigure (d) shows the rank-2 projective measurement: the ancilla pair $S_A,S_B$ is connected to the target nodes by two star branches (cyan and orange), and a joint Pauli-$X$ measurement on the pair concurrently establishes all logical edges between the two branch sets.}
    \label{fig:network_primitives}
\end{figure*}

\subsection{Quantum network model and basic evolution operations}
\label{subsec:network_model}

We adopt a quantum network physical model similar to that of Ref.~\cite{refMMG}. As shown in Fig.~\ref{fig:network_primitives}(a), the quantum network is described by an undirected graph $\mathcal{N} = (V_\mathcal{N}, E_\mathcal{N})$. A dashed circle represents a physical node, which may contain multiple qubits internally. A double solid line $e \in E_\mathcal{N}$ represents a quantum channel connecting neighboring nodes.
At the physical resource layer, the basic resource of the network is the Bell pair spanning a neighboring channel $e$. A Bell state can be converted into an edge between two nodes by a local Hadamard gate; below, we directly regard Bell states as the underlying entanglement resource for graph state edges. In this paper, the byproduct operators induced by Pauli measurements can be decomposed into an outcome-independent local Clifford layer and an outcome-dependent Pauli layer~\cite{refgraph0,refgraph1}. The former is determined at compile time and is applied as deterministic single-qubit operations scheduled with the circuit, without classical feedforward. The latter consists of Pauli operators, which can be tracked and corrected via classical feedforward~\cite{refpauli1}. Under the time-slot model adopted in this paper~\cite{refshiqian,refhuang}, classical communication results arrive within a single time slot, and the tracking and settlement of the Pauli frame occupy no extra time slots. The above deterministic single-qubit corrections are not counted in the metrics of this paper.
We also assume that the physical layer, via entanglement purification and other low-level protocols, can provide high-fidelity neighboring Bell pairs on demand. On this basis, we construct three basic operations for graph state evolution.
\begin{itemize}
    \item \textbf{Long-range entanglement.} As shown in Fig.~\ref{fig:network_primitives}(b), when a remote target node needs an entanglement connection, CZ gates and Pauli-$Y$ measurements are executed sequentially along the shortest path, consuming the Bell pairs on the path and cascading the short-range Bell pairs into a long-range entanglement edge.
    \item \textbf{State stitching.} As shown in Fig.~\ref{fig:network_primitives}(c), a long-range entanglement edge is attached to an existing graph state via a local CZ gate, and the intermediate qubit is then removed by a Pauli-$Y$ measurement, so that the new edge merges seamlessly into the target topology.
    \item \textbf{Rank-2 projective measurement.} After completing the physical connections of the ancilla pair $(S_A, S_B)$ according to the vectors $u, v$ using the above operations, a joint Pauli-$X$ measurement is performed on the pair. This operation is the rank-2 algebraic projection at the physical level. The concrete execution of this operation is illustrated in Fig.~\ref{fig:network_primitives}(d).
\end{itemize}

To illustrate the physical execution of Algorithm~\ref{alg:algebraic} more intuitively, Fig.~\ref{fig:toy_model} shows an example with a 6-particle target graph state. The rank of the adjacency matrix of this target graph state is $\operatorname{rank}(M_{\text{target}})=4$; by Theorem~\ref{thm:complexity}, only $\lfloor 4/2\rfloor = 2$ joint Pauli-$X$ measurement steps are needed to complete the synthesis. The underlying entanglement establishment and topology evolution are shown in the figure.

Figure~\ref{fig:toy_model} demonstrates the key feature of our method: a single joint Pauli-$X$ measurement establishes multiple entanglement edges concurrently in one shot, with the number of edges determined by the vectors $u$ and $v$. This method also reuses entanglement resources efficiently. Unlike the scheme that distributes star subgraphs one by one along Steiner trees, one rank-2 update can concurrently cover multiple target edges. Redundant edges, shown as red lines in the figure, may appear during synthesis; these edges do not belong to the target graph state. Since $1\oplus 1=0$ over $\mathrm{GF}(2)$, all redundant edges are naturally cancelled by subsequent updates in later steps. The total number of steps of the whole process is strictly bounded by Theorem~\ref{thm:complexity} and is independent of the density of the target graph state.

\begin{figure*}[t]
    \centering
    \includegraphics[width=0.8\textwidth]{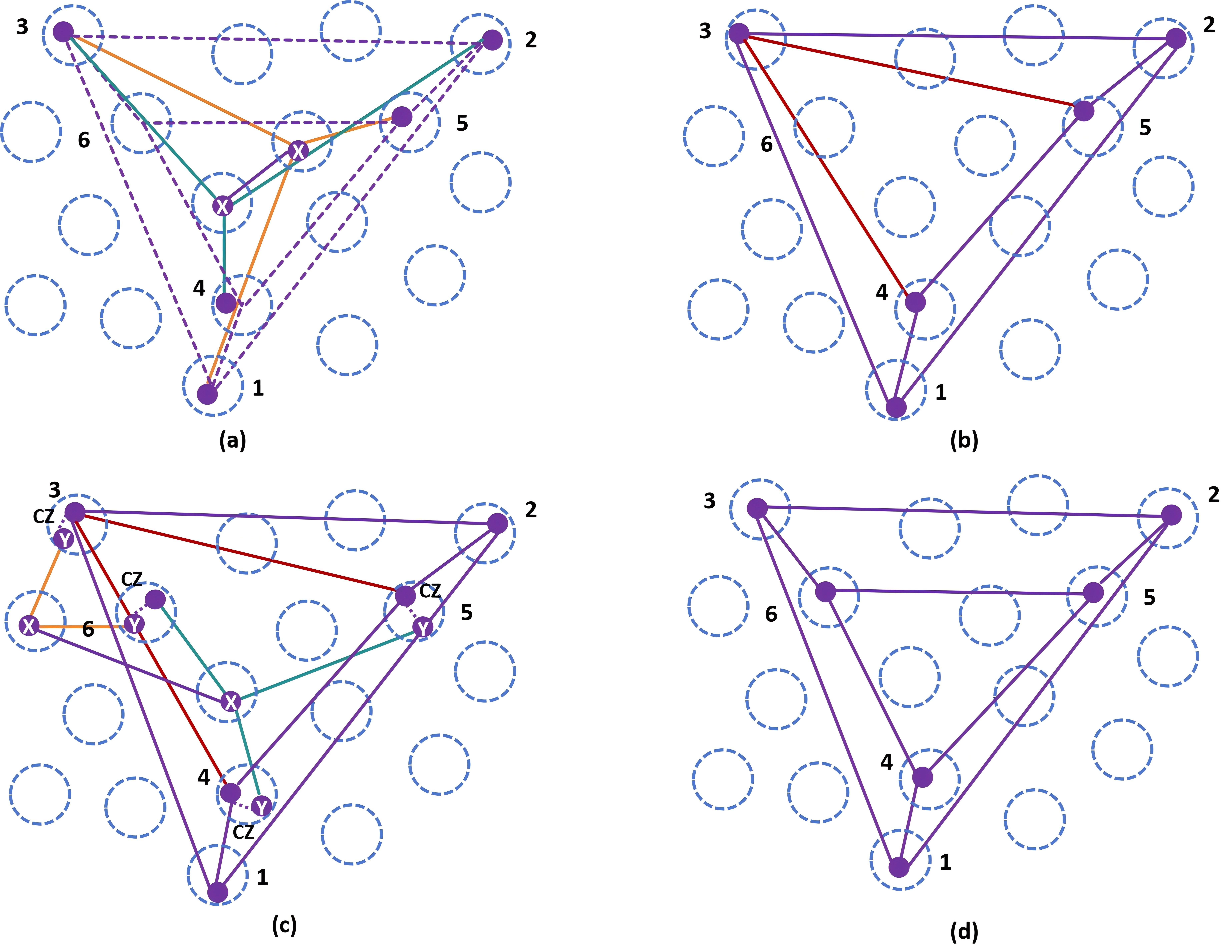}
    \caption{Synthesis example of a 6-node target state. In subfigure (a), $R_{1,2}=1$ is selected; physical entanglement is established according to the difference matrix and the first joint Pauli-$X$ measurement is executed; the cyan and orange lines are the two star branches. Subfigure (b) shows the logical edges generated concurrently after the measurement; purple edges are target edges and red edges are redundant edges. In subfigure (c), a new pivot $R_{3,4}=1$ is selected; after state stitching via CZ gates and Pauli-$Y$ measurements, the second joint measurement is executed. In subfigure (d), the redundant red edges are automatically eliminated, and the purple synthesized graph agrees with the target state shown by the gray dashed lines.}
    \label{fig:toy_model}
\end{figure*}

\subsection{Single-step depth of a rank-2 update}
\label{sec:concurrent}

As described above, to implement a joint Pauli-$X$ measurement at the physical level, a specific entanglement structure must first be established according to the vectors $u$ and $v$. Below, we first define this structure and then define its degree-based single-step depth measure. Let $S_A^{(k)},S_B^{(k)}$ denote the ancilla pair used in step $k$.

\begin{definition}[Dual-star concurrent distribution]
\label{def:dual_star}
For a rank-$2$ update $\Delta R = uv^{\mathrm T}+vu^{\mathrm T}$, the dual-star concurrent distribution refers to the ancilla pair $(S_A,S_B)$ executing three branches simultaneously on the physical network. These three branches are the backbone entanglement between $S_A$ and $S_B$, the star-type entanglement from $S_A$ to the set $\{x:u_x=1\}$, and the star-type entanglement from $S_B$ to the set $\{y:v_y=1\}$.
\end{definition}

Under the time-slot model adopted in this paper, the network operates synchronously by slots. Quantum nodes satisfy the single-interface constraint: within one slot, a node can participate in the establishment of only one entanglement link. Therefore, multiple entanglement links on the same node must be scheduled into different slots and executed serially, while established entanglement is stored by the node and maintained across slots. 
Let $G_{\mathrm{phys}}^{(k)}$ denote the routing graph formed by the physical edges traversed by the three branches of step $k$; 
For the numerical comparison, we use the maximum node degree of this graph as a degree-based slot-depth measure:

\begin{equation}
\label{eq:concurrent_delay}
\widetilde{\mathrm{Depth}}^{(k)} = \Delta\bigl(G_{\mathrm{phys}}^{(k)}\bigr) = \max_{w \in G_{\mathrm{phys}}^{(k)}} \deg(w).
\end{equation}

The maximum node degree $\Delta$ provides a natural degree-based measure
of the single-step scheduling cost. The busiest node must participate in
$\Delta$ edges in that step and can handle only one edge per slot, so
$\Delta$ is a lower bound on the exact slot count. For a general routing
graph, additional scheduling conflicts may require extra slots. We use
$\Delta$ consistently as the slot-depth measure in the numerical comparison.
For the Steiner-tree baseline, this measure coincides with the exact
edge-scheduling depth.

The degree-based time-slot depth is the sum over all steps. On sparse physical topologies with fixed average degree, star distributions all unfold into Steiner trees at the physical layer, and the maximum node degree of the tree is constrained by the physical connectivity, growing slowly with the target density $p$ and quickly saturating. 
Our dual-star distribution and the baseline single-star distribution show no essential difference in order, although merging three branches makes ours somewhat higher. When single-step depths are of the same order, the difference falls on the number of steps. The baseline step count equals the number of decomposed star subgraphs, approaching $N-1$ as the density increases. Our step count does not exceed $\lfloor N/2\rfloor$ and is independent of the density. We can therefore expect our advantage to be larger for denser targets. The numerical experiments in Sec.~\ref{sec:numerical} will verify this.

\subsection{Heuristic algorithm under physical mapping}
\label{subsec:heuristic}
The algebraic algorithm of Sec.~\ref{sec:3} is based on the ideal Bell pair assumption. In real networks, long-distance entanglement is affected by channel attenuation and gate noise; it consumes a large number of raw Bell pairs and requires multiple rounds of entanglement purification to distill entangled pairs reaching the fidelity threshold~\cite{refbriegel1998}. The number of purification rounds grows rapidly with distance and hop count~\cite{refchen,refzhao2022e2e}. For the same algebraic operation, different ancilla positions and different pivot elimination orders can lead to vastly different physical costs.

We model the underlying entanglement resources as follows. Each physical channel independently prepares and buffers short-range Bell pairs between neighboring nodes, and long-range entanglement is obtained by sequentially performing entanglement swapping on the buffered Bell pairs~\cite{refbriegel1998,refpant2019routing}. The resource cost of an edge $e=(a,b)$ is taken as the dominant term of its purification overhead, $C_e=\exp(\alpha_{\mathrm{attn}}L_e)$, where $L_e=\|p_a-p_b\|$ is the Euclidean length and $\alpha_{\mathrm{attn}}$ is the fiber attenuation constant.
One step of dual-star concurrent distribution needs to establish three groups of entanglement links at the physical layer, connecting $S_A$ to the set $\{x:u_x=1\}$, $S_B$ to the set $\{y:v_y=1\}$, and the backbone link between $S_A$ and $S_B$. The star links are laid along a Steiner tree covering all leaf nodes, the backbone is laid along the shortest path, and the three are merged into the routing graph $G_{\mathrm{phys}}^{(k)}$ of that step. Only one entangled pair needs to be established on each physical edge of the routing graph; intermediate nodes cascade them into long-range connections via entanglement swapping, and the total cost equals the sum of edge costs of the routing graph.

After mapping to the physical network, the choice of a suitable pivot $(i,j)$ and the placement of the ancilla pair $(S_A,S_B)$ both affect the total resource overhead. Since a rank-2 update introduces redundant edges while eliminating target edges, and although these redundant edges are cancelled in subsequent steps, they all incur distribution costs physically. Therefore, let $R$ be the current difference matrix and $|E_R|=\tfrac{1}{2}\|R\|_1$ the total number of target edges not yet established. After one step with pivot $(i,j)$, the difference matrix becomes $R_{\mathrm{next}}$, and we define the net progress of edge synthesis:
\begin{equation}
\label{eq:net_progress}
\Delta_{\mathrm{net}}(i,j) = |E_R| - |E_{R_{\mathrm{next}}}| = \tfrac{1}{2}\bigl(\|R\|_1 - \|R_{\mathrm{next}}\|_1\bigr).
\end{equation}
The larger $\Delta_{\mathrm{net}}$, the more target edges are eliminated on net in that step. The more redundant edges, the lower this value is pulled.

A natural approach would be to jointly optimize the pivot $(i,j)$ and the placement of the ancilla pair $(S_A,S_B)$, minimizing the per-progress cost $C^{(k)}/\Delta_{\mathrm{net}}$ of each step. However, joint optimization means that a complete ancilla placement search must be rerun for every candidate pivot, adding a factor of $|E_R|$ to the complexity compared with placement alone. Note that $\Delta_{\mathrm{net}}$ is a pure $\mathrm{GF}(2)$ quantity, determined only by the difference matrix and independent of the physical position of the ancilla pair in the network. Meanwhile, the placement of the ancilla pair $(S_A,S_B)$ can be solved exactly once the pivot $(i,j)$ is fixed and $u$ and $v$ are determined. Therefore, we can replace the joint optimization by a two-stage optimization: in the first stage, among all candidate pivots satisfying $R_{i,j}=1$, we choose the one with the largest $\Delta_{\mathrm{net}}$; in the second stage, we select the minimum-cost ancilla pair for this pivot only. Each of the two subproblems is still solved optimally on its own. The pivot ordering in the first stage adopts a greedy strategy, which is the only greedy component of the algorithm.
For the ancilla placement in the second stage, the most direct criterion would be to solve the Steiner tree cost of the routing graph exactly for each candidate; however, there are up to $O(|\mathcal{V}_{\mathrm{ancilla}}|^2)$ candidate pairs, and constructing a Steiner tree for each is too expensive. In fact, we can use a shortest-path approximation: with $C_e$ as edge weights, precompute the shortest paths between all node pairs, and let $E_{a,b}$ denote the shortest-path edge set from $a$ to $b$. Take the union of the shortest-path edge sets of the three branches as the approximate edge set of step $k$,
\begin{equation}
\label{eq:edge_union}
E^{(k)} = E_{S_A S_B} \cup \Bigl(\bigcup_{x:\,u_x=1} E_{S_A,x}\Bigr) \cup \Bigl(\bigcup_{y:\,v_y=1} E_{S_B,y}\Bigr),
\end{equation}
and the sum of its edge costs is the approximate cost of the candidate ancilla pair,
\begin{equation}
\label{eq:cost_func}
C^{(k)}(S_A, S_B; u, v) = \sum_{e\in E^{(k)}} C_e .
\end{equation}

The cost given by the shortest-path union is an upper bound on the optimal Steiner tree cost; for the 2-approximate Steiner tree actually used, it is not a strict upper bound.
Under the same pivot, all candidates face the same leaf set, and the degree of merging of common path segments mainly depends on the spatial distribution of the leaves, so the deviation between the approximate cost and the true cost is relatively stable across candidates; ranking by the approximate cost thus selects an ancilla pair close to optimal. The approximation is used only for ranking; all performance metrics are ultimately counted on the actual routing graphs.
The optimal ancilla pair is obtained by exhaustive search over the candidate set, with complexity $O(|\mathcal{V}_{\mathrm{ancilla}}|^2)$ edge-set unions, which is a polynomial-time exact search. 
Our model does not impose additional capacity or edge-disjointness constraints on the concurrent branches, so the ancilla-placement search defined above remains polynomial in the number of candidate ancilla nodes. The complete procedure is given in Algorithm~\ref{alg:heuristic}.

\begin{algorithm}[htbp]
\AlgCaption{Heuristic algorithm under physical mapping}
\label{alg:heuristic}
\begin{algorithmic}[1]
\REQUIRE Target adjacency matrix $M_{\mathrm{target}}$; all-pairs shortest paths and their edge sets precomputed with edge weights $C_e=\exp(\alpha_{\mathrm{attn}}L_e)$; ancilla node set $\mathcal{V}_{\mathrm{ancilla}}$; physical topology $G_{\mathrm{phys}}$
\ENSURE Operation sequence $\mathcal{S}$, total physical cost $\mathrm{Cost}_{\mathrm{tot}}$
\STATE $R \leftarrow M_{\mathrm{target}}$, $\mathrm{Cost}_{\mathrm{tot}} \leftarrow 0$, $\mathcal{S} \leftarrow \emptyset$
\WHILE{$R \neq \mathbf{0}$}
    \STATE Among all candidate pivots with $R_{i,j}=1$, choose $(i,j)$ with the largest $\Delta_{\mathrm{net}}$ \COMMENT{pure $\mathrm{GF}(2)$ criterion}
    \STATE $u \leftarrow Re_j$, $v \leftarrow Re_i$
    \STATE Exhaustively search the ancilla pair $(S_A,S_B)$ over $\mathcal{V}_{\mathrm{ancilla}}$, compute the approximate cost by Eqs.~(\ref{eq:edge_union})--(\ref{eq:cost_func}), and take the minimum \COMMENT{complexity $O(|\mathcal{V}_{\mathrm{ancilla}}|^2)$}
    \STATE Execute the dual-star concurrent distribution with the selected $(i,j,S_A,S_B)$ and complete the joint Pauli-$X$ measurement; count the physical cost $W^{(k)}$ of this step on the actual routing graph; record the operation $\mathcal{O}^\ast$
    \STATE $R \leftarrow R \oplus (uv^{\mathrm{T}}+vu^{\mathrm{T}}) $
    \STATE $\mathrm{Cost}_{\mathrm{tot}} \leftarrow \mathrm{Cost}_{\mathrm{tot}} + W^{(k)}$, $\mathcal{S} \leftarrow \mathcal{S}\cup\{\mathcal{O}^\ast\}$
\ENDWHILE
\RETURN $\mathcal{S}$, $\mathrm{Cost}_{\mathrm{tot}}$
\end{algorithmic}
\end{algorithm}

\section{Numerical experiments and result analysis}
\label{sec:numerical}
In this section, we conduct Monte Carlo experiments on physical topologies with fiber attenuation, using the Steiner baseline as the control, and evaluate the performance along three dimensions: time, consumed resources, and gate operation counts.

\subsection{Baseline and experimental setup}
\label{sec:setup}

For the baseline, we take Ref.~\cite{refMMG} as the skeleton and strengthen it where necessary; that is, the target graph state is greedily decomposed into star subgraphs, each time taking the node with the largest degree in the residual graph together with its unfinished neighboring edges to form a star. The star subgraph in each step is generated using the lower-gate-cost method of Refs.~\cite{refckl,refscalable}, and the entanglement establishment follows the Steiner tree idea. The slot count within a step is the maximum node degree of the routing graph, and the steps are strictly serial. In each experiment, the two protocols are executed on the same physical graph and target graph instances.

In the code implementation, the Steiner tree is solved by \texttt{steiner\_tree} in NetworkX, i.e., the Mehlhorn 2-approximation algorithm~\cite{refmehlhorn}. The root cause of the baseline time increasing with density is the growth in the number of star subgraphs from the greedy decomposition, reaching $N-1$ stars for a complete graph; this trend is independent of the approximation accuracy of the tree, so the core conclusion in the time dimension is not affected by the approximation.
Gates and measurements are counted by the following rules~\cite{refckl,refscalable}. Any node of degree $d\ge 2$ in the routing structure needs to locally stitch a star state of $d$ qubits, counted as $d-1$ CZ gates. Leaf nodes of degree $1$ incur no gates. A logical node belonging to the final graph state keeps one qubit as the physical carrier and only measures out the remaining $d-1$ ancilla qubits. Pure routing nodes do not enter the final graph state and must measure out all $d$ qubits. The ancilla pair $(S_A,S_B)$ is counted as pure routing nodes, and their joint Pauli-$X$ measurement is already included in the above $d$ measurements and is not counted again.

The physical network is randomly generated by the Waxman model~\cite{refwaxman1988routing}. Nodes are sampled in the unit square with a minimum separation enforced. Candidate edges are generated with probability $P(a,b)=\beta_w\exp(-d(a,b)/(\alpha_w L))$, with $L=\sqrt{2}$ and $\alpha_w=0.10$, and candidate edges longer than $3.5\,\alpha_w L$ are removed. For each node number $N_{\mathrm{phys}}$, $\beta_w$ is adjusted so that the average degree is fixed at $\langle k\rangle=3$, ensuring that $\alpha_w$ only affects the spatial distribution pattern of edges without changing the total number of edges. If the generated graph is disconnected, edges are added one by one between connected components at the Euclidean-nearest node pairs to bridge them. The resource cost of each physical edge is counted by fiber attenuation, $C_e=\exp(\alpha_{\mathrm{attn}}\cdot L_e)$, with $\alpha_{\mathrm{attn}}=0.5$.
The target graph states are ER random graphs $G(N,p)$~\cite{referdds1959random,referdHos1960evolution}; if disconnected, edges are added between connected components one by one to make them connected. The parameter grid is $N\in[10,90]$ (step 10) and $p\in[0.1,0.9]$ (step 0.1). The physical scale is $N_{\mathrm{phys}}=3N$. The $N$ target nodes are randomly chosen from the physical nodes, and the ancilla pair of our scheme is chosen from the remaining $2N$ nodes. For each $(N,p)$ combination, 300 independent experiments are performed. 
Time is evaluated using the degree-based time-slot model of Sec.~\ref{sec:concurrent}:
the slot-depth measure of each step is the maximum node degree of the
routing graph, and the steps are accumulated serially.

\begin{figure*}[t]
\centering
\includegraphics[width=0.9\textwidth]{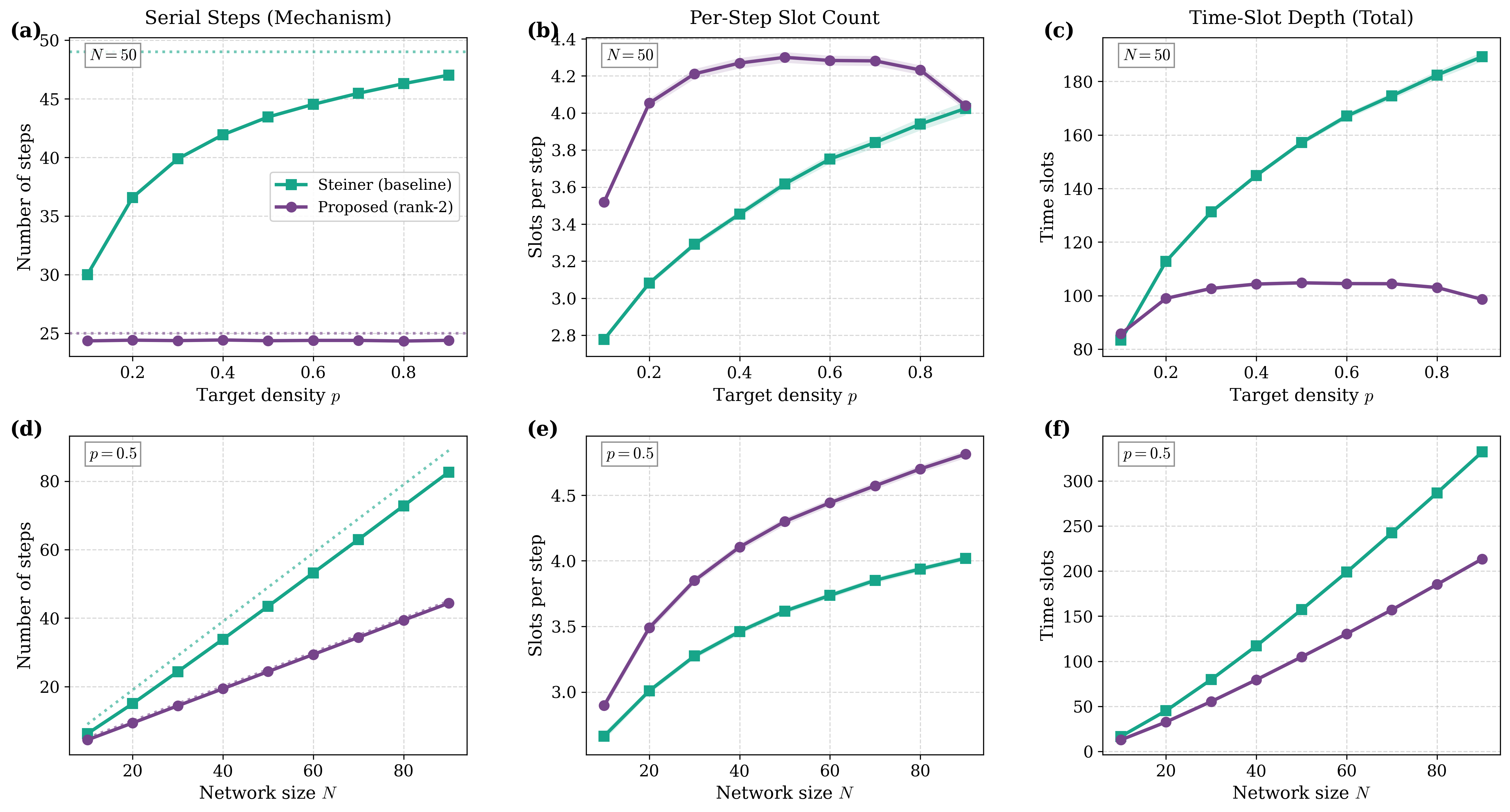}
\caption{Time performance and mechanism comparison. The physical topology has a fixed average degree $\langle k \rangle = 3$; each data point is taken from 300 Monte Carlo experiments, and the shaded bands are $95\%$ confidence intervals. Subfigures (a) and (d) show the number of serial steps, (b) and (e) the average degree-based slot-depth measure per step, and (c) and (f) the total degree-based time-slot depth. Our step count remains stable near $\lfloor N/2\rfloor$ and does not vary with $p$, while the baseline step count grows with $p$ toward $N-1$; the per-step counts are of the same order for both protocols and grow only slowly, so the total depth is mainly determined by the number of serial steps, and the time advantage of our protocol keeps expanding with density and scale. The dotted lines in (a) and (d) are theoretical references: $N-1$ is the worst-case step count of the baseline, and $\lfloor N/2 \rfloor$ is the step-count upper bound of our method given by Theorem~\ref{thm:complexity}.}

\label{fig:time_perf}
\end{figure*}

\subsection{Time cost}
\label{subsec:time_cost}

We first analyze the performance of the protocols in terms of time cost.
Figure~\ref{fig:time_perf} compares the number of serial steps of our protocol with the baseline. As shown in subfigure (a), the step count of our method remains stable near $N/2$, while the baseline step count grows with $p$. This is because the Steiner baseline step count equals the number of decomposed star subgraphs: as the target graph density $p$ increases, the step count gradually approaches $N-1$. In contrast, the step count of our method is always determined by the rank of the adjacency matrix of the target graph state, upper bounded by $\lfloor N/2\rfloor$, and is not affected by $p$.

Meanwhile, since the average degree of our physical topology is set to $\langle k \rangle = 3$, the maximum node degrees of the two protocols are similar, ours being slightly higher because the dual-star routing graph merges three branches. Subfigures (b) and (e) confirm this: the per-step slot counts of both protocols remain of the same order and grow only slowly, whereas the step counts grow linearly. Consequently, the total time-slot depth of the two schemes is mainly determined by the number of serial steps.
Subfigure (c) demonstrates this. At $p=0.1$, the time-slot depths of the two protocols are basically the same. As $p$ increases, the baseline's time-slot depth rises nearly linearly, while our method's time-slot depth, after a period of slow growth, forms a plateau around $p\approx0.4$--$0.6$ and then falls back. The gap grows larger and larger. Subfigure (d) shows the same step-count comparison versus $N$ at fixed $p=0.5$. Subfigure (f) shows the variation of the time-slot depth with the network size $N$ at fixed $p=0.5$; the growth rate of our method is again clearly lower than that of the baseline. This demonstrates the good scalability of the protocol, and the time advantage of the protocol becomes more pronounced as the density increases.

\begin{figure*}[htbp]
    \centering
    \includegraphics[width=\textwidth]{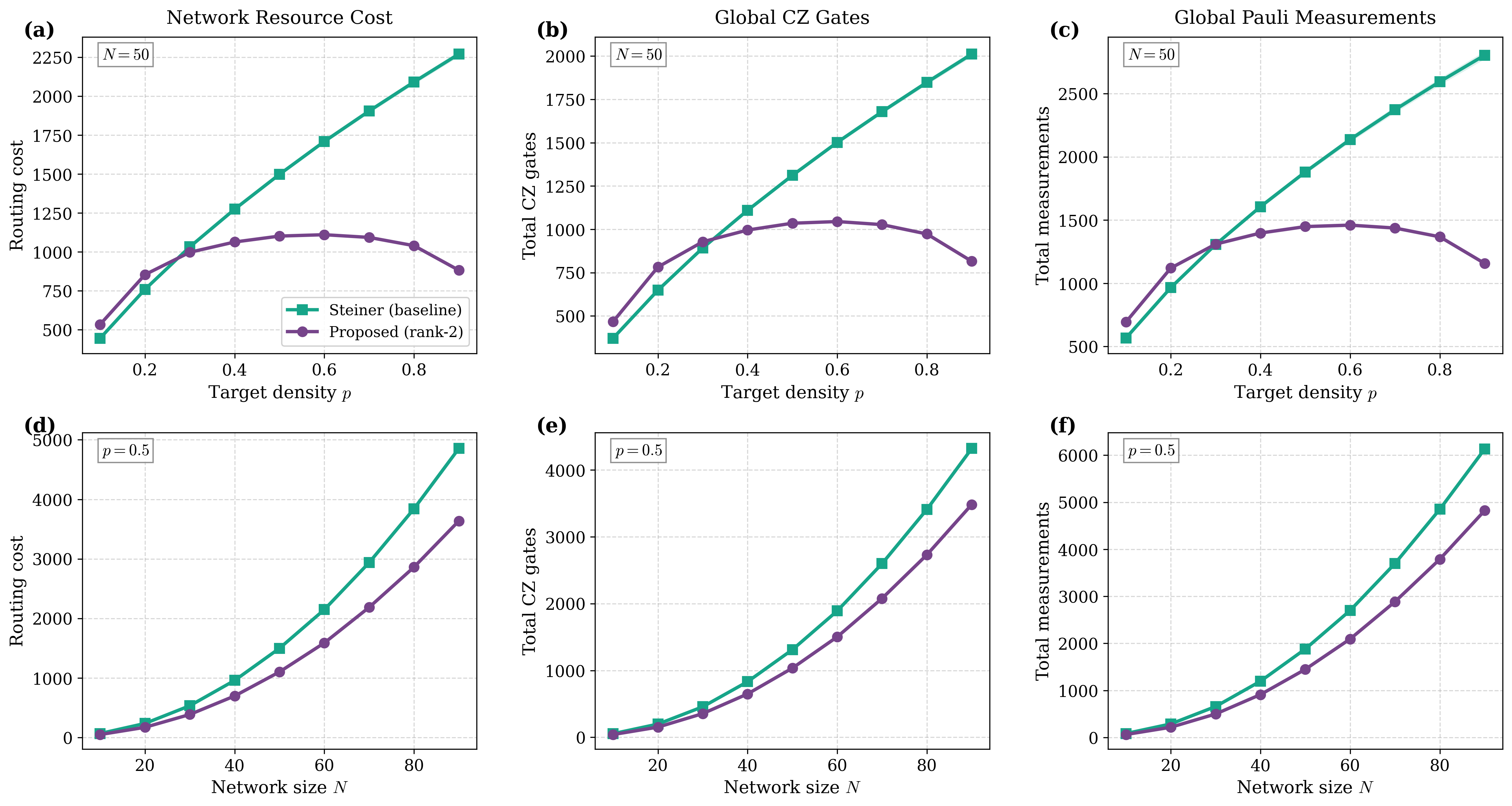}
    \caption{Global overhead comparison. The three columns show the entanglement resource overhead, the total number of CZ gates, and the total number of Pauli measurements, respectively. The upper row fixes $N=50$ and varies the target density $p$; the lower row fixes $p=0.5$ and varies the network size $N$. In the low-density region, the baseline is slightly better in resources and operation counts; beyond $p\approx0.3$, our method comprehensively overtakes the baseline, and the higher the density, the more obvious the advantage; as $N$ increases, both protocols maintain superlinear growth.}
    \label{fig:overheads}
\end{figure*}

\subsection{Resource consumption and quantum operation overhead}
\label{subsec:overheads}

Figure~\ref{fig:overheads} shows the trends of the underlying entanglement resources consumed by the two protocols, the total number of CZ gates, and the total number of Pauli measurements as functions of the target graph state edge density $p$ and node number $N$. As can be seen from the figure, the trends of the three metrics are similar, indicating that the underlying cause is the same.
Specifically, subfigures (a--c) show that, in the low-density region $p \lesssim 0.3$, the Steiner baseline is lower than our method in both entanglement resources and quantum operation overhead. This is because the Steiner tree method itself is designed with minimizing the entanglement resource overhead as the optimization objective, whereas our algebraic synthesis method must establish a dual-star concurrent distribution at every step, so in the low-density region its resource utilization efficiency is lower than the baseline.
However, as $p$ increases, each rank-2 update of our algebraic synthesis method generates more target entanglement edges simultaneously, the entanglement resources of the dual-star concurrent distribution are highly reused, and the construction cost per edge decreases. The Steiner baseline, on the other hand, has to rebuild a Steiner tree for every star it distributes, and its resources, CZ gates, and measurement operations grow with the number of star subgraphs. As a result, beyond $p>0.3$, our method comprehensively overtakes the baseline in entanglement resources and quantum operation overhead, and the higher the density, the more obvious the advantage.
Subfigures (d)--(f) likewise show the variation of the corresponding overheads of the two protocols with the node number $N$ at fixed density $p=0.5$. The trends of the two protocols are similar, both maintaining superlinear growth as $N$ increases.

In summary, the experimental results show that our algebraic graph state synthesis method is, in the time dimension, a scheme with fewer steps and lower time-slot depth than the Steiner baseline. In terms of entanglement resource utilization efficiency and quantum operation consumption, our method is close to the baseline and shows an advantage in the dense region.

\section{Conclusion and outlook}
\label{sec:conclusion}

This paper transforms the graph state synthesis problem into a process of successively reducing the rank of the target graph state adjacency matrix over $\mathrm{GF}(2)$, and proves that the number of steps required to synthesize an arbitrary $N$-node graph state is $\operatorname{rank}(M_{\mathrm{target}})/2$, no more than $\lfloor N/2\rfloor$, with this upper bound independent of the target graph density.
At the physical implementation level, each joint Pauli-$X$ measurement step is mapped to a dual-star concurrent distribution.
We modeled the protocol on Waxman topologies with fiber attenuation, gave a heuristic algorithm, and conducted Monte Carlo numerical experiments against the strengthened Steiner baseline.
The experimental results show that, thanks to the gap between $\lfloor N/2\rfloor$ and the baseline step count $(N-1)$, our method is superior in time-slot depth almost everywhere, except in the extremely sparse region ($p\approx0.1$) where it is basically on par with the baseline. The entanglement resource overhead, total number of CZ gates, and total number of quantum measurements also become significantly better beyond $p\approx0.3$, and the denser the target graph state, the more obvious the advantage.

Looking ahead, we believe several directions are worth further study. First, Algorithm~\ref{alg:heuristic} is currently a heuristic placement based on resource cost; whether better approximation algorithms exist is a question worth further investigation. Second, the experimental results show that our method has clear advantages in the high-density region; we can further study its applications in preparing specific dense graph states such as complete multipartite graph states. Third, the number of preparation steps of our method is determined by the rank of the adjacency matrix of the target graph state, which means that, by LC preprocessing of the target graph state, i.e., using LC operations to find an equivalent graph state with a smaller rank, the number of synthesis steps could be further reduced; related questions also deserve deeper study. Finally, regarding the generality of our method: although this work is developed within the framework of synthesizing graph states in distributed networks, the underlying algebraic elimination logic applies to all graph-state-based systems, and the transfer of this method to other scenarios is also a direction for future exploration.

\subsection*{Authors' contributions}
X.Z. conceived the main approach, developed the theoretical framework and protocol, carried out the numerical simulations, and prepared the initial manuscript. L.C. provided the initial research direction and conceptual guidance, supervised the project, and revised the manuscript. C.-T.L. contributed to supervision, discussions, and manuscript review. All authors discussed the results and approved the final manuscript.
Large language models were used only for grammar correction and language polishing, with all edits reviewed by the authors.

\section*{Data Availability}
The simulation scripts and numerical data that support the findings of this study are openly available at: \url{https://github.com/ky87zheng/graph-state-gf2}.

\bibliographystyle{quantum}
\bibliography{references}

\end{document}